\documentclass[12pt,a4paper]{amsart}

\usepackage[utf8]{inputenc}
\usepackage[osf,noBBpl]{mathpazo}
\usepackage{microtype}
\usepackage{xcolor}
\definecolor{darkgreen}{rgb}{0,0.6,0}

\usepackage[pdftex,bookmarks=true,bookmarksnumbered=true,pdfpagemode=UseOutlines,pdfstartview=FitH,%
    pdfpagelayout=OneColumn,colorlinks=true,urlcolor=blue,citecolor=darkgreen,pdfborder={0 0 0},ocgcolorlinks]{hyperref}

\usepackage{amsmath,amsfonts,amssymb,amsthm}
\usepackage{mathtools}

\usepackage{geometry}

\usepackage{algorithm}
\usepackage{algpseudocode}
\algrenewcommand\algorithmicrequire{\textbf{Input:}}
\algrenewcommand\algorithmicensure{\textbf{Output:}}
\newcommand\Input\Require
\newcommand\Output\Ensure

\theoremstyle{plain}
\newtheorem{thm}{Theorem}[section]
\newtheorem{lem}[thm]{Lemma}
\newtheorem{prop}[thm]{Proposition}
\newtheorem{cor}[thm]{Corollary}
\newtheorem*{ostrowski}{Ostrowski's Theorem}
\newtheorem*{newtonpuiseux}{Newton-Puiseux Theorem}
\theoremstyle{definition}
\newtheorem{defn}[thm]{Definition}

\newcommand\refsec[1]{Sec.~\ref{sec:#1}}

\newcommand\KK{\mathbb K}
\newcommand\LL{\mathbb L}
\newcommand\QQ{\mathbb Q}
\newcommand\ZZ{\mathbb Z}
\newcommand\NN{\mathbb N}
\newcommand\NP{\mathsf{NP}}
\newcommand\poly{\mathsf{poly}}
\newcommand\I{\mathsf{M}}
\newcommand\bigO{O}
\newcommand\Scal{\mathcal S}
\newcommand\Fcal{\mathcal F}
\newcommand\Hcal{\mathcal H}
\newcommand\closure\overline
\DeclareMathOperator\mult{mult}
\newcommand\puiseux[2][\closure\KK]{#1\langle\!\langle#2\rangle\!\rangle}

\DeclareMathOperator\val{val}
\DeclareMathOperator\mval{mval}
\DeclareMathOperator\mdeg{mdeg}
\DeclareMathOperator\Newt{Newt}
\DeclareMathOperator\res{res}
\renewcommand\wr{\operatorname{wr}}
\DeclareMathOperator\size{size}

\newcommand\proj[1]{#1_\pi}
\newcommand\tup[1]{\mathbold{#1}}
\newcommand\ii{\tup{i}}
\newcommand\XX{\tup{X}}
\renewcommand\aa{\tup{\alpha}}
\newcommand\bb{\tup{\beta}}
\newcommand\dd{\tup{\delta}}
\newcommand\PolyRing{\KK[\XX]}

\usepackage[numbers]{natbib}

\let\citep\cite

\title{Bounded-degree factors of lacunary multivariate polynomials}
\author{Bruno Grenet}
\address{LIRMM -- Université de Montpellier\\ UMR 5506 CNRS}
\email{Bruno.Grenet@lirmm.fr}

\begin{document}
\maketitle

\begin{abstract}
In this paper, we present a new method for computing bounded-degree factors of lacunary multivariate polynomials. In particular for polynomials over number fields, we give a new algorithm that takes as input a multivariate polynomial $f$ in lacunary representation and a degree bound $d$ and computes the irreducible factors of degree at most $d$ of $f$ in time polynomial in the lacunary size of $f$ and in $d$. Our algorithm, which is valid for any field of zero characteristic, is based on a new gap theorem that enables reducing the problem to several instances of (a) the univariate case and (b) low-degree multivariate factorization.

The reduction algorithms we propose are elementary in that they only manipulate the exponent vectors of the input polynomial. The proof of correctness and the complexity bounds rely on the Newton polytope of the polynomial, where the underlying valued field consists of Puiseux series in a single variable.
\end{abstract}

\section{Introduction} 

The factorization of polynomials is a well-studied subject in symbolic computation. Although there exist effective fields in which testing irreducibility of polynomials is undecidable~\citep{FroShe55}, the irreducible factorization of univariate or multivariate polynomials can be computed in time polynomial in the degree of the input polynomial for many base fields. Without claim of exhaustiveness, one can cite the cases of polynomials over rational numbers~\citep{LeLeLo82,Kal89} and algebraic number fields~\citep{Len83,Lan85,Len87}, or over finite fields~\citep{Ber67}. From a somewhat different perspective, one can also compute the factorization in an extension of the base field, such as (approximate) factorization in the real or complex numbers~\citep{Pan02,KaMayYaZhi08} or absolute factorization, that is factorization over an algebraic closure of the base field~\citep{CheGa05}.

The purpose of this paper is to propose polynomial-time algorithms when the input polynomial is given in \emph{lacunary representation}, that is as a list of nonzero monomials. These algorithms have complexity \emph{logarithmic} in the degree.\footnote{The lacunary representation is also known as sparse representation in the literature. Yet is customary to use the term \emph{lacunary} for algorithms of complexity logarithmic in the degree, and \emph{sparse} for algorithms of complexity polynomial in the degree.} Note that in lacunary representation, even evaluating a polynomial over an input is intractable: For instance, the monomial $X^d$ has lacunary size $\bigO(\log d)$ while its evaluation on the input $2$ is an integer of size $d$. More generally, testing the irreducibility of lacunary polynomials or computing the greatest common divisor of two lacunary polynomials are $\NP$-hard problems~\citep{Pla77,KaShp99,KaKoi05}. This motivates refining our ambitions and computing only a partial factorization of the input polynomial, namely the irreducible factors of bounded degree.

\subsection{Previous work}

\citet*{CuKoiSma99} gave an algorithm to compute the integer roots of univariate integer polynomials in time polynomial in the lacunary representation. This result was generalized by \citet*{Len99} who described an algorithm to compute the bounded-degree factors of polynomials over number fields. His algorithm takes as input a description of the number field by means of an irreducible polynomial with integer coefficients in dense representation, the polynomial to factor in lacunary representation, and a bound on the degree of the factors it computes. The complexity is polynomial in the size of the input and in the degree bound (rather than in its bit-size). Then, \citet*{KaKoi05} generalized this result to the computation of linear factors of bivariate polynomials over the rational numbers, and then to the computation of bounded-degree factors of multivariate polynomials over number fields~\citep{KaKoi06}. Seemingly independently of this latest result, \citet*{AvKriSo07} generalized the first result of \citet*{KaKoi05} and gave an algorithm to compute the bounded-degree factors of bivariate polynomials over number fields. They also explained how to compute the bounded-degree factors with at least three monomials over an algebraic closure of the rational numbers. Note that the binomial factors include univariate linear factors and the number of such factors cannot be polynomially bounded in the logarithm of the degree. We proposed another algorithm for the computation of the multilinear factors in the bivariate and multivariate cases~\citep{ChaGreKoiPoStr13,ChaGreKoiPoStr14}. Since it relies on Lenstra's algorithm for univariate factors, it is valid in full generality over number fields only, though our approach works in more general settings and allow for partial results over any fields of characteristic zero and to some extent in positive characteristic. All these results are based on a technique, due to \citet*{CuKoiSma99}, that consists in finding \emph{gaps} in the input polynomial (\emph{cf.} next section). 

\citet*{Ave09} proposed a different technique to test whether a given linear factor divides a lacunary bivariate polynomial, again over number fields. To our knowledge, his approach does not allow to compute the factors. It is based on a bound on the number of real roots of the intersection of a lacunary polynomial with a line. This latter result has been extended to the intersection of a lacunary polynomial with a low-degree polynomial by \citet*{KoiPoTa15}. It appears that Avenda\~no's method could be combined with this more recent result to obtain an algorithm that tests whether a given low-degree polynomial divides a lacunary bivariate polynomial. Nevertheless this algorithm would only work with some low-degree polynomials, since it requires in particular the polynomial to have real roots.

Let us finally mention two other results. \citet*{Sag14} gave an algorithm to compute the real roots of an integer polynomial with arithmetic complexity polynomial in the size of the lacunary representation of the input but with exponential bit-complexity. \citet*{BiCheRo13} proposed an algorithm to compute the roots of lacunary polynomials over finite fields that run in sublinear time in the degree, and proved that this problem is $\NP$-hard (under randomized reductions).

\subsection{Main results and techniques}

The algorithms of this paper are generalizations of our algorithms for computing multilinear factors of bivariate and multivariate polynomials~\citep{ChaGreKoiPoStr13,ChaGreKoiPoStr14}. We identify two distinct kinds of factors, namely the \emph{unidimensional} and the \emph{multidimensional} factors. Roughly speaking, a polynomial is said unidimensional if it can be written as $f(\XX^\dd)$ where $f$ is a univariate polynomial and $\XX^\dd$ a multivariate monomial. 
We describe an algorithm to reduce the computation of the bounded-degree unidimensional factors of a lacunary multivariate polynomial to the computation of the bounded-degree factors of some lacunary univariate polynomials. It is based on the fact that unidimensional factors of a multivariate polynomial are in bijection with the irreducible factors of some \emph{univariate projections} of the polynomial. It is valid for any base field, in any characteristic. Though this paper focuses on bounded-degree unidimensional factors, the reduction is more general and could be used as well for low-degree polynomials, or for the computation of the lacunary unidimensional factors for instance.
For the multidimensional factors, we give an algorithm to reduce their computation to the irreducible factorization of a low-degree polynomial. This algorithm is based on a so-called \emph{Gap Theorem}, valid for any field of characteristic $0$, that asserts that if a polynomial $f$ can be written $f_1+\XX^\dd f_2$ where $\dd$ is large enough, each low-degree factor of a $f$ is a common factor of $f_1$ and $f_2$. 
Both algorithms are elementary since they only manipulate the exponent vector of $f$ and not its coefficients. 

The proof of our Gap Theorem is based on the notion of Puiseux expansion of a bivariate polynomial, and makes use of the Newton polygon of the polynomial. Given a polynomial $g\in\KK[X,Y]$, one can describe its \emph{roots} in an algebraic closure $\puiseux X$ of $\KK(X)$ in terms of Puiseux series, that is formal power series whose exponents are rational numbers with a common denominator. We give a bound on the valuation of $f(X,\phi(X))$ where $f$ is a lacunary polynomial and $\phi\in\puiseux X$ cancels a low-degree polynomial $g$. It only depends on the degree of $g$ and the number of nonzero monomials of $f$. Its proof is based on the wronskian determinant of a family of linearly independent power series.

As a corollary, we obtain a new proof of the main result of \citet*{KaKoi06} stating that over algebraic number fields, one can compute the degree-$d$ factors of a lacunary multivariate polynomial $f$ in deterministic time $(\size(f)+d)^{\bigO(n)}$ or probabilistic time $(\size(f)+d)^{\bigO(1)}$, where $\size(f)$ is the lacunary size of $f$. Their algorithm uses a universal constant arising from number theory that is not explicitly given. In contrast, our algorithm is entirely explicit and can easily be implemented~\citep{Gre15}. 

Since our Gap Theorem applies to any field of characteristic $0$, we obtain partial results for other fields. In particular, for any field that admits a multivariate factorization algorithm running in time polynomial in the degree of the input polynomial, we obtain a algorithm to compute the bounded-degree multidimensional factors of multivariate polynomials that runs in time polynomial in the lacunary size of the input polynomial and the degree bound. Such fields include the fields of real or complex numbers, the fields of $p$-adic numbers or the algebraic closure of the rational number. In this latest case, we obtain a new proof of \citet*[Theorem~3.5]{AvKriSo07} since the unidimensional irreducible polynomials over $\closure\QQ$ are exactly the binomials. 

\subsection{Open questions}

Our results leave open some questions. First, our Gap Theorem does not apply as such for fields of positive characteristic. Yet, in the specific case of multilinear factors and fields of large characteristic, we proved that it is applicable, yielding an algorithm~\citep{ChaGreKoiPoStr13}. We suspect that a similar result holds for low-degree factors more generally, though we were not able to prove it yet. For fields of small characteristic, the same gap argument does not seem to apply but there may well exist a different approach that exploits the fact that the characteristic is small. 

Another question concerns the factors computed. As mentioned above, we cannot hope for a polynomial-time algorithm computing the full irreducible factorization of a lacunary polynomial, whence the restriction on the degree of the factors. A natural generalization would be to impose a bound on the number of nonzero monomials of the factors instead of their degree. Our reduction for unidimensional factors is valid in this context but it is not usable since an algorithm for the univariate case misses. Since this problem concerns lacunary polynomials, the possibility of a hardness result should not be excluded though. 

\subsection*{Organization of the paper}

In \refsec{prelim}, we introduce the necessary notions used throughout the paper. \refsec{unidim} is devoted to the unidimensional factors, and \refsec{multidim} to the multidimensional factors. 

\subsection*{Note} A preliminary version of this article appeared in the proceedings of ISSAC~\citep{Gre14}. The current paper is a complete rewriting of the preceding article. In particular, the algorithms for multivariate polynomials were only sketched. In this version, the algorithms have been simplified and are described in full details. Moreover, we tighten our complexity analyses.

\subsection*{Acknowledgements} I would like to thank again P. Koiran, N. Portier, Y. Strozecki, A.
Bostan, P. Lairez and B. Salvy for discussions that helped me to prepare the preliminary version of this paper. I would also like to thank R. Lebreton for discussions on this new version, and an anonymous reviewer for very interesting remarks and suggestions.

\section{Preliminaries} \label{sec:prelim} 

\subsection{Notations} 

Let $\XX$ denote the tuple of variables $(X_1,\dotsc,X_n)$, and $\aa$ the tuple $(\alpha_1,\dotsc,\alpha_n)$. The notation $\XX^\aa$ denotes the monomial $\prod_{i=1}^n X_i^{\alpha_i}$. For any scalar $k$, $k\aa$ denotes the vector $(k\alpha_1,\dotsc, k\alpha_n)$. 

We consider polynomials in a ring $\PolyRing=\KK[X_1,\dotsc,X_n]$. A polynomial of $\PolyRing$ with $k$ (nonzero) terms is called an \emph{$n$-variate $k$-nomial}. The degree of a polynomial $f$ with respect to the variable $X_i$ is denoted $\deg_{X_i}(f)$ and its valuation with respect to $X_i$ is denoted $\val_{X_i}(f)$. Let also $\mdeg(f)$ denotes the multidegree $(\deg_{X_1}(f),\dotsc,\deg_{X_n}(f))$ and $\mval(f)$ the multivaluation $(\val_{X_1}(f),\dotsc,\val_{X_n}(f))$. If $f$ is a univariate polynomial, we shall use the usual notations $\deg(f)$ and $\val(f)$ for its degree and valuation. 
We say that a polynomial has multidegree \emph{at most} $(d_1,\dotsc,d_n)$ if $\deg_{X_i}(f)\le d_i$ for all $i$. 

The multiplicity of $g$ as a factor of $f$, that is the maximum integer $\mu$ such that $g^\mu$ divides $f$, is denoted by $\mult_g(f)$.

A \emph{partition} of a polynomial $f=\sum_{j=1}^k c_j \XX^{\aa_j}$ is a set of polynomials $\{f_1,\dotsc,f_s\}$ defined by a partition of $\{1,\dotsc,k\}$ into disjoint subsets. The polynomials $f_1$, \dots, $f_s$ are the \emph{summands} of the partition. In particular, $f=f_1+\dotsb+f_s$ and two distinct summands do not share any common monomial. We shall mainly write partitions as sums $f_1+\dotsb+f_s$ rather than as sets $\{f_1,\dotsc,f_s\}$. 

Let $\size(f)$ denote the \emph{lacunary size} of a polynomial: If $f=\sum_{j=1}^k c_j \XX^{\aa_j}\in\PolyRing$ is an $n$-variate $k$-nomial, 
\[\size(f)=k(n\max_{\substack{1\le i\le n\\1\le j\le k}}(\log(\alpha_{i,j}))+\size(c_j))\]
where the size of $c_j$ depends on the field $\KK$. Actually since our algorithms only manipulate the exponent vectors, we shall express the complexities as functions of $k$, $n$ and the total degree $D$ of $f$. More precisely, we aim to describe algorithms of complexity polynomial in $k$, $n$ and $\log D$. 
We shall also use the notion of \emph{convex size} which denotes the volume of its Newton polytope (\emph{cf.} next section). 

To express the complexities of some of our algorithms, we shall use the notation $\I(m)$ which denotes the complexity of the multiplication of two $m$-bit integers. It satisfies $\I(m)=\bigO(m\log m 8^{\log^\star m})$~\citep{Fur09,HaHoeLe14} where $\log^\star$ is the iterated logarithm, recursively defined by $\log^\star(1)=0$ and $\log^\star(m) = 1 + \log^\star(\log m)$. Note for instance that computing the greatest common divisor of two $m$-bit integers takes $\bigO(\I(m)\log m)$ bit-operations.

\subsection{Newton polytope} 

\begin{defn}
Let $f\in\PolyRing$. Its \emph{support} is the set of vectors $\aa$ such that the coefficient of the monomial $\XX^\aa$ in $f$ is nonzero. The \emph{Newton polytope} of $f$, denoted by $\Newt(f)$, is the convex hull of its support.
\end{defn}

A polytope in two dimensions is called a \emph{polygon}, whence the appellation \emph{Newton polygon} when $f$ is bivariate.

Two convex polytopes can be added using the \emph{Minkowski sum}, defined by $A+B=\{a+b : a\in A, b\in B\}$ for two polytopes $A$ and $B$. If $A$ and $B$ are both convex, then so is $A+B$. Minkowski sum is related to factorization of polynomials by Ostrowski's Theorem.

\begin{ostrowski}
Let $f$, $g$, $h\in\PolyRing$ such that $f=gh$. Then $\Newt(f)=\Newt(g)+\Newt(h)$.
\end{ostrowski}

The boundary of the Newton polygon of a bivariate polynomial is made of edges whose extremities are points of the support. In $n$ variables, the boundary of a Newton polytope is made of faces which have dimension $1$ to $(n-1)$. We shall only consider faces of dimension $1$ that we still call edges. Again, the extremities of an edge are points in the support of the polynomial. 
We define the \emph{direction of an edge} of extremities $\aa$ and $\bb$ as the unique vector $\dd\in\ZZ^n$ collinear to $\aa-\bb=(\alpha_1-\beta_1,\dotsc,\alpha_n-\beta_n)$ whose first nonzero coordinate is positive and such that $\gcd(\delta_1,\dotsc,\delta_n)=1$.  

Ostrowski's Theorem shall mainly be used through a corollary concerning the edges of the Newton polytopes.
\begin{cor}\label{cor:ostrowski}
Let $f$, $g$, $h\in\PolyRing$ such that $f=gh$. Then each edge in $\Newt(f)$ is parallel to either an edge of $\Newt(g)$ or an edge of $\Newt(h)$.
\end{cor}

For some algorithms, we may want to have access to the Newton polytope of $f$. Actually, this is doable in polynomial time for a fixed number of variables only~\citep{BergCheongKreveldOvermars}. 

\begin{prop}\label{prop:convexhull}
Let $\aa_1$, \dots, $\aa_k\in\ZZ^n$, $\|\aa_1\|_\infty$, \dots, $\|\aa_k\|_\infty\le D$. Their convex hull can be deterministically computed in $\bigO(k^{\lfloor n/2\rfloor}n\I(\log D))$ bit-operations.
\end{prop}

\subsection{Puiseux series} 

\begin{defn}
Let $\closure\KK$ be an algebraic closure of a field $\KK$ of characteristic $0$. The \emph{field of Puiseux series over $\closure\KK$}, denoted by $\puiseux X$, is defined as the set of formal power series in a single variable of the form
\[\phi(X) = \sum_{t\ge t_0} f_t X^{t/d}\]
where $t_0\in\ZZ$, $d\in\NN$, and $f_t\in\closure\KK$ for all $t$.

If $f_{t_0}\neq 0$, the \emph{valuation} of $\phi$ is $\val(\phi) = t_0/d$.
\end{defn}

Addition and multiplication are defined as for standard formal power series. A theorem of Puiseux states that Puiseux series actually form a field, and that this field is algebraically closed. We can give a constructive version of this result using the Newton polygon of a bivariate polynomial $f\in\KK[X,Y]$. 
Since $\Newt(f)$ lives in this case in a plane, one can choose a system of coordinates to draw it
and define the lower and upper parts of the Newton polygon: To each edge, let us attach a normal vector pointing inside the Newton polygon; The \emph{lower hull} is defined as the set of edges whose normal vectors have a positive second coordinate, the \emph{upper hull} as the set of edges whose normal vectors have a negative second coordinate, and \emph{vertical edges} are the (at most two) edges whose normal vectors have their second coordinate equal to zero. The names are clear if we choose to represent the exponents of $Y$ on the $x$-axis and the exponents of $X$ on the $y$-axis.

\begin{newtonpuiseux}
Let $g\in\KK[X,Y]$. There exists $\phi\in\puiseux X$ of valuation $v$ such that $g(X,\phi(X))=0$ if and only if the lower hull of $\Newt(g)$ contains an edge of direction $(p,q)$ such that $v=-p/q$.
\end{newtonpuiseux}

If $g$ has bidegree $(d_X,d_Y)$, any edge of its Newton polygon is contained in the rectangle $(0,d_X)\times (0,d_Y)$, hence has direction $(p,q)$ with $|p|\le d_X$ and $|q|\le d_Y$. As a consequence, the valuation $v$ of any root $\phi\in\puiseux X$ of $g$ satisfy $1/d_Y\le |v|\le d_X$. 

\subsection{The wronskian determinant} 

Our main technical result uses the wronskian determinant of a family of Puiseux series.

\begin{defn}
Let $(\phi_1,\dotsc,\phi_\ell)$ be a family of Puiseux series. Its \emph{wronskian} is
\[\wr(\phi_1,\dotsc,\phi_\ell) = \det \begin{pmatrix}
    \phi_1 & \dotsb & \phi_\ell\\
    \phi_1' & \dotsb & \phi_\ell'\\
    \vdots & & \vdots\\
    \phi_1^{(\ell-1)} & \dotsb & \phi_\ell^{(\ell-1)}
    \end{pmatrix}\text.\]
\end{defn}

The main property we shall use is the relation of the wronskian to the linear independence. \citet*{BoDu10} give a proof of the following proposition in the case of formal power series. The exact same proof applies to Puiseux series.

\begin{prop}
Let $\phi_1$, \dots, $\phi_\ell\in\puiseux X$ be Puiseux series. The family $(\phi_j)_{1\le j\le\ell}$ is linearly independent if and only if its wronskian does not vanish.
\end{prop}

We aim to give bounds on the valuation of the wronskian of certain particular families of Puiseux series in \refsec{valuation}. We can first give a general lower bound.

\begin{lem}\label{lem:valLowerBound}
Let $\phi_1$, \dots, $\phi_\ell\in\puiseux X$. Then 
\[\val(\wr(\phi_1,\dotsc,\phi_k)) \ge \sum_{j=1}^\ell \val(\phi_j) -\binom{\ell}{2}\text.\]
\end{lem}

\begin{proof} 
Using the full symbolic expansion of the determinant, $\wr(\phi_1,\dotsc,\phi_k)$ can be written as a sum of terms of the form
$\phi_1^{(i_1)}\dotsc \phi_\ell^{(i_\ell)}$
such that $\{i_1,\dotsc,i_\ell\}=\{0,\dotsc,\ell-1\}$. 
These terms have valuations
\[ \sum_{j=1}^\ell\val(\phi_j)-i_j = \sum_{j=1}^\ell\val(\phi_j)-\sum_{j=0}^{\ell-1} j= \sum_{j=1}^\ell\val(\phi_j)-\binom{\ell}{2}\text.\]
The valuation of the wronskian is at least as large.
\end{proof} 

\subsection{Uni- and multidimensional polynomials} 

In our algorithms, we treat in very different ways unidimensional and multidimensional polynomials. We first give a formal definition of these terms.

\begin{defn}
A polynomial $f\in\PolyRing$ is \emph{unidimensional} if the dimension of its Newton polytope is exactly $1$, that is if $f$ has at least two monomials and its Newton polytope is a line segment. The \emph{direction} $\dd\in\ZZ^n$ of $f$ is the direction of the unique edge of its Newton polytope.

A polynomial is \emph{multidimensional} if its Newton polytope has dimension at least $2$.
\end{defn}

Note that monomials are neither unidimensional nor multidimensional. Since the computation of the monomial factors is obvious, we ignore them in the rest of the paper.

We remark that for bivariate polynomials, being unidimensional is the same as being weighted-homogeneous. This is not true anymore for polynomials in more variables. We now define several notions, by analogy with homogeneous polynomials.

For $f\in\PolyRing$ and $\dd\in\ZZ^n$, one can write $f=f_1+\dotsc+f_s$ where each $f_t$ is either unidimensional of direction $\dd$ or a monomial. If further no sum $f_{t_1}+f_{t_2}$ is unidimensional (that is the $f_t$'s are maximal), the $f_t$'s are the \emph{unidimensional components of direction $\dd$} of $f$, in short its \emph{$\dd$-components}.

Homogenization and dehomogenization are called projection and lifting in our settings. Let us first prove a lemma to justify the definitions.

\begin{lem} 
Let $f\in\PolyRing$ be a unidimensional polynomial of direction $\dd$. There exists a unique univariate polynomial $\proj f\in\KK[Z]$ of valuation $0$ such that $f(\XX)=\XX^\aa \proj f(\XX^\dd)$ for some $\aa\in\ZZ^n$. Furthermore, $\aa$ is in this case nonnegative, that is $\aa\in\NN^n$.
\end{lem}

\begin{proof} 
Let $f(\XX)=\sum_{j=1}^\ell c_j \XX^{\aa_j}$. Since $f$ is unidimensional of direction $\dd$, there exists for all $j$ an integer $\lambda_j$ such that $\aa_j-\aa_1=\lambda_j \dd$. Let $j_0$ be the index of the smallest $\lambda_j$ and $\lambda_j^\star=\lambda_j-\lambda_{j_0}$ for all $j$. The integers $\lambda_1^\star$, \dots, $\lambda_\ell^\star$ are nonnegative and satisfy $\aa_j-\aa_{j_0}=\lambda_j^\star\dd$ for all $j$. 

Let us define $f_\pi = \sum_{j=1}^\ell c_j Z^{\lambda_j^\star}$. Then $f_\pi$ belongs to $\KK[Z]$ and has valuation $0$. Moreover if we let $\aa=\aa_{j_0}$, we have
\[\XX^\aa f_\pi(\XX^\dd)=\sum_{j=1}^\ell c_j \XX^{\aa+\lambda_j^\star\dd}=\sum_{j=1}^\ell c_j \XX^{\aa_j}\]
since $\aa+\lambda_j^\star\dd=\aa_j$ by definition of $\lambda_j^\star$. This proves the existence of $f_\pi$.

To prove its uniqueness, let us consider $g=\sum_{p=1}^\ell b_p Z^{\gamma_p}$ with $0=\gamma_1<\dotsb<\gamma_\ell$ such that $f(\XX)=\XX^\aa g(\XX^\dd)$ for some $\aa\in\ZZ^n$. Clearly, since $f$ is a polynomial and $g$ has valuation $0$, $\aa$ belongs to $\NN^n$ and is the exponent of some monomial of $f$. Now, for all $j$ there exists a $p$ such that the term $c_j\XX^{\aa_j}$ of $f$ is the image of the term $b_pZ^{\gamma_p}$ of $g$. In particular, $\aa_j-\aa=\gamma_p\dd$. Since $\gamma_p\ge 0$, the uniqueness of the index $j_0$ defined in the first paragraph shows that $\aa=\aa_{j_0}$. Each $\gamma_p$ is therefore uniquely defined by the differences $(\aa_j-\aa_{j_0})$ and $g=\proj f$.
\end{proof} 

From this lemma, one can define \emph{the} projection of a unidimensional polynomial.

\begin{defn}
Let $f\in\PolyRing$ a unidimensional polynomial of direction $\dd$. Its \emph{(univariate) projection} is the unique polynomial of lowest degree $\proj f\in\KK[Z]$ such that $f(\XX)=\XX^\aa \proj f(\XX^\dd)$ for some $\aa\in\NN^n$

Let $g\in\KK[Z]$ a univariate polynomial and $\dd\in\ZZ^n$. Its \emph{lifting in direction $\dd$} is the unique unidimensional polynomial $f\in\PolyRing$ of multivaluation $\tup 0$ defined by $f(\XX)=\XX^\aa g(\XX^\dd)$ for some $\aa\in\ZZ^n$. 
\end{defn}

Note that if $f\in\PolyRing$ is a unidimensional polynomial of valuation zero with respect to each of its variables, the operations of projection and lifting are inverse of each other. That is, the lifting in direction $\dd$ of $\proj f$ is $f$ itself. In general, the lifting of $f_\pi$ is the polynomial $f^\circ$ defined by 
$f^\circ(\XX)=f(\XX)/\XX^{\mval(f)}$. In other words, if two unidimensional polynomials $f$ and $g$ of direction $\dd$ have the same projection, there exists $\aa\in\ZZ^n$ such that $f(\XX)=\XX^\aa g(\XX)$.

We shall need a bound on the degree of the projection $\proj g$ of a unidimensional polynomial $g$. Let us assume that $\mval(g)=\tup 0$, for example that $g$ is irreducible, and let $\dd$ be the direction of $g$ and $\tup d$ its multidegree. By definition, there exists $\aa$ such that $g(\XX)=\XX^\aa\proj g(\XX^\dd)$. Consider an index $i$ such that $\delta_i\neq0$. If $\delta_i>0$, then $d_i=\alpha_i+\delta_i \deg(\proj g)$, and since $g$ and $\proj g$ have valuation $0$, $\alpha_i=0$. Thus $\deg(\proj g)=d_i/\delta_i$. If $\delta_i<0$, we get $\alpha_i=d_i$ and $0=\alpha_i+\delta_i\deg(\proj g)$. Whence in both cases $\deg(\proj g)=d_i/|\delta_i|$. In particular, let us assume that we only a bound on the multidegree of $g$, then $\deg(\proj g)=\min_i (d_i/|\delta_i|)$ where the minimum is taken over all the indices $i$ such that $\delta_i\neq0$.

\section{Unidimensional factors} \label{sec:unidim} 

In this section, we show how to reduce the computation of the unidimensional factors of some polynomial $f\in\PolyRing$ to the factorization of (several) univariate polynomials. 

\subsection{Structural result} 

\begin{thm} \label{thm:projcomp}
Let $f\in\PolyRing$ and $\dd\in\ZZ^n$. Let $f_1$, \dots, $f_s$ its $\dd$-components and $\proj{(f_1)}$, \dots, $\proj{(f_s)}\in\KK[Z]$ their respective projections. For any unidimensional polynomial $g$ of direction $\dd$,
\[\mult_g(f) = \min_{1\le t\le s} \mult_{\proj g}(\proj{(f_t)})\]
where $\proj g$ is the projection of $g$.
\end{thm}

This theorem is a direct consequence of the two following lemmas. 

\begin{lem} \label{lem:components}
Let $f\in\PolyRing$, and $\dd\in\ZZ^n$. The unidimensional factors of direction $\dd$ of $f$ are the common factors of its $\dd$-components. More precisely, 
\[\mult_g(f) = \min_{1\le t\le s} \mult_g(f_t)\]
where $f_1$, \dots, $f_t$ are the $\dd$-components of $f$.
\end{lem}

\begin{proof} 
The product of two unidimensional polynomials $g$ and $h$ of direction $\dd$ is itself a unidimensional polynomial of direction $\dd$. Indeed let $f=gh$ and consider two monomials $\XX^\aa$ and $\XX^\bb$ of $f$. Each monomial is a product of a monomial of $g$ and a monomial of $h$. Let us assume that $\XX^\aa=\XX^{\aa_g}\XX^{\aa_h}$ and $\XX^\bb=\XX^{\bb_g}\XX^{\bb_h}$ where $\XX^{\aa_g}$ and $\XX^{\bb_g}$ are monomials of $g$ and $\XX^{\aa_h}$ and $\XX^{\bb_h}$ are monomials of $h$. Then
\begin{align*} 
\aa-\bb &=(\aa_g+\aa_h)-(\bb_g+\bb_h)=(\aa_g-\bb_g)+(\aa_h-\bb_h)\\
        &= \lambda_g\dd+\lambda_h\dd=(\lambda_g+\lambda_h)\dd
\end{align*}
for some $\lambda_g$, $\lambda_h\in\ZZ$. This shows that $f$ is unidimensional of direction $\dd$.

Consider now a unidimensional factor $g$ of direction $\dd$ of some polynomial $f\in\PolyRing$, and let $h=f/g$. Let us write $h=h_1+\dotsb+h_s$ as a sum of $\dd$-components. Then $gh = g h_1+\dotsb+gh_s$ and each $gh_t$ is unidimensional of direction $\dd$. This proves in particular that $g$ divides each $\dd$-component of $f$. To conclude, it remains to notice that the same argument works for $g^\mu$ where $\mu=\mult_g(f)$.
\end{proof} 

\begin{lem} 
Let $f$ and $g\in\PolyRing$ be unidimensional polynomials of same direction and $\proj f$ and $\proj g$ their respective projections. Then
\[\mult_g(f)=\mult_{\proj g}(\proj f)\text.\]
\end{lem}

\begin{proof} 
Let us first prove that if $f=gh$ is a unidimensional polynomial, its projection is $\proj g\proj h$ where $\proj g$ and $\proj h$ are the respective projections of $g$ and $h$. Let $\dd$ be the direction of $f$, $g$ and $h$. By definition, there exist $\aa_g$ and $\aa_h\in\ZZ^n$ such that $g(\XX)=\XX^{\aa_g}\proj g(\XX^\dd)$ and $h(\XX)=\XX^{\aa_h}\proj h(\XX^\dd)$. Thus, $f(\XX)=\XX^{\aa_g+\aa_h} \proj g(\XX^\dd)\proj h(\XX^\dd)$. Let $\proj f=\proj g\proj h$. Then $f(\XX)=\XX^{\aa_g+\aa_h}\proj f(\XX^d)$ and $\proj f$ is the projection of $f$, by uniqueness of the projection.

Let us assume that $g^\mu$ divides $f$ for some $\mu>0$, that is there exists $h$ such that $f=g^\mu h$. The projection of $f$ is $\proj f=\proj g^\mu\proj h$, and $\mult_{\proj g}(\proj f)\ge \mu$. Conversely, let us assume that $\proj f=\proj g^\mu\proj h$ for some $\proj g$, $\proj h\in\KK[Z]$, and denote by $g$ and $h$ the respective liftings of $\proj g$ and $\proj h$ in direction $\dd$. Let $f^\circ= g^\mu h$, so that $\proj f^\circ=\proj g^\mu\proj h=\proj f$. There exists $\aa\in\ZZ^n$ such that $f^\circ(\XX)=\XX^\aa f(\XX)$. Since $g$ and $h$ are prime with $\XX^\aa$, $g^\mu$ divides $f$ and $\mult_g(f)\ge\mu$. This concludes the proof.
\end{proof} 

\subsection{Computing the set of directions} \label{sec:directions} 

The goal of this section is to compute, given $f\in\PolyRing$, the set $\Delta_0(f)\subset\ZZ^n$ of directions $\dd$ such that $f$ has a unidimensional factor of direction $\dd$. More precisely, we are going to compute an approximation of $\Delta_0(f)$, that is a set $\Delta$ that contains $\Delta_0(f)$. We give several algorithms with distinct and often incomparable complexities, that compute different approximations of $\Delta_0(f)$.

Let $f$, $g\in\PolyRing$ such that $g$ is unidimensional of direction $\dd$ and divides $f$. Corollary~\ref{cor:ostrowski} implies that the Newton polytope of $f$ has two parallel edges of direction $\dd$. More precisely, $\Newt(f)$ can in this case be partitioned into line segments of direction $\dd$, none of which in reduced to a single point. This motivates the definition of three supersets of $\Delta_0(f)$.

\begin{defn}\label{def:delta}
Let $f\in\PolyRing$. The three sets $\Delta_1(f)$, $\Delta_2(f)$, $\Delta_3(f)\subset\ZZ^n$ are defined as follows:
\begin{itemize} 
\item $\dd\in\Delta_1(f)$ if the support of $f$ can be partitioned into line segments of direction $\dd$, none of which is reduced to a single point;
\item $\dd\in\Delta_2(f)$ if $\Newt(f)$ has two parallel edges of direction $\dd$;
\item $\dd\in\Delta_3(f)$ if the support of $f$ has two points $\aa$, $\bb$ such that $\aa-\bb$ has direction $\dd$.
\end{itemize}
\end{defn}

\begin{lem}\label{lem:delta}
Let $f\in\PolyRing$ be an $n$-variate $k$-nomial. Then
\[\Delta_0(f)\subseteq\Delta_1(f)\subseteq\Delta_2(f)\subseteq\Delta_3(f)\]
and $|\Delta_3(f)|\le\binom{k}{2}$.
\end{lem}

\begin{proof} 
The first two inclusions were proved above, and follow from Corollary~\ref{cor:ostrowski}. The last inclusion simply comes from the fact that an edge of $\Newt(f)$ connects two points of the support of $f$. The bound on the cardinality of $\Delta_3(f)$ follows from the same observation since there are at most $\binom{k}{2}$ pairs of points in the support of $f$.
\end{proof} 

We now give algorithms to compute $\Delta_1(f)$, $\Delta_2(f)$ and $\Delta_3(f)$, beginning with the easiest to compute.

\begin{lem}\label{lem:delta3}
Given an $n$-variate $k$-nomial $f\in\PolyRing$ of total degree $D$, one can compute $\Delta_3(f)$ in $\bigO(k^2n\I(\log D)\log D)$ bit-operations.
\end{lem}

\begin{proof} 
The algorithm is straightforward: For each pair of exponent vectors $\{\aa,\bb\}$ of $f$, one computes the direction of $(\aa-\bb)$. There are $\binom{k}{2}$ such pairs to consider. Computing the direction of a vector $(\aa-\bb)$ reduces to some arithmetic operations and gcd computations on $\bigO(n)$ integers of size at most $\log D$. 

It remains to detect collisions in the directions. One can to this end sort the directions using some total order, say lexicographic. This can be computed in $\bigO(k^2n\log D)$ operations using Radix Sort, since there are $\bigO(k^2)$ vectors of $n$ integers of size at most $\log D$. The conclusion follows since $\log D\le\I(\log D)$.
\end{proof} 

We turn to the computation of $\Delta_2(f)$. This can be done without computing first $\Delta_3(f)$ in order to avoid considering all the $\binom{k}{2}$ pairs of points. Note though that in the worst case, the quadratic dependence on $k$ is unavoidable since $\Delta_2(f)$ may have $\bigO(k^2)$ edges.

\begin{lem}\label{lem:delta2}
Given an $n$-variate $k$-nomial $f\in\PolyRing$ of total degree $D$, one can compute $\Delta_2(f)$ in $\bigO(k^{\lfloor n/2\rfloor}n\I(\log D)+k^2n\I(\log D)\log D)$ bit-operations.
\end{lem}

\begin{proof} 
One can use Proposition~\ref{prop:convexhull} to compute the Newton polytope of $f$. The output of such an algorithm gives a list of facets, from which one can extract the edges. We simply have the return the set $\Delta_2(f)$ of directions $\dd$ such that there are two distinct edges of direction $\dd$ in $\Newt(f)$. The cost is $\bigO(k^{\lfloor n/2\rfloor}n\I(\log D))$ to compute the Newton polytope, and $\bigO(k^2n\I(\log D)\log D)$ to compute the directions of the at most $\binom{k}{2}$ edges. 
\end{proof} 

Computing $\Delta_2(f)$ is thus expensive. Though, if $n=3$ for instance, $\Newt(f)$ cannot have more than $\bigO(k)$ edges and the cost become linear in $k$.

We now turn to the computation of $\Delta_1(f)$. We propose two approaches. In the first one, one computes $\Delta_3(f)$ (or $\Delta_2(f)$) and extracts $\Delta_1(f)$ from it by removing the directions $\dd$ such that the support of $f$ cannot be partitioned into line segments of direction $\dd$. The second one is direct.

\begin{lem} \label{lem:lines}
Let $\aa_1$, \dots, $\aa_k$ and $\dd\in\ZZ^n$, $\|\aa_1\|_\infty,\dotsc,\|\aa_k\|_\infty,\|\dd\|_\infty\le D$. One can compute a partition of $\{\aa_1,\dotsc,\aa_k\}$ into line segments of direction $\dd$ in time $\bigO(nk\I(\log D))$.
\end{lem}

\begin{proof} 
Let $\Hcal$ denote the hyperplane whose normal vector is $\dd$. Two points $\aa_i$ and $\aa_j$ belong to the same line segment of direction $\dd$ if and only if their projections onto $\Hcal$ coincide. The projections can be computed using a dot product: The projection of $\aa$ onto $\Hcal$ is given by
\[\aa-\frac{(\aa\cdot \dd)}{\|\dd\|_2^2}\dd\] 
where $(\aa\cdot\dd)=\sum_i\alpha_i\delta_i$ is the dot (or scalar) product of $\aa$ and $\dd$. This projection can be computed in time $\bigO(n\I(\log D))$ using arithmetic operations on the coordinates of the vectors. Therefore, one can compute the projection of each $\aa_j$ onto $\Hcal$ in time $\bigO(nk\I(\log D))$. 
It remains to detect the collisions between these projections as in Lemma~\ref{lem:delta3}.
Altogether, a partition of $\{\aa_1,\dotsc,\aa_k\}$ into line segments can be computed in $\bigO(nk\I(\log D))$ bit-operations.
\end{proof} 

\begin{lem}\label{lem:delta1}
Given an $n$-variate $k$-nomial $f\in\PolyRing$ of total degree $D$, one can compute $\Delta_1(f)$ in $\bigO(k^2n^2\I(\log D)\log D)$ bit-operations.
\end{lem}

\begin{proof} 
As mentioned before, the first strategy is to use Lemma~\ref{lem:delta3} to compute $\Delta_3(f)$ and then for each $\dd\in\Delta_3(f)$, to check whether the support of $f$ can be partitioned into line segments of direction $\dd$ using Lemma~\ref{lem:lines}. This takes $\bigO(k^2n\I(\log D)\log D+k^3n\I(\log D))$ bit-operations. Similarly, computing first $\Delta_2(f)$ using Lemma~\ref{lem:delta2} and refining it to obtain $\Delta_1(f)$ takes $\bigO(k^{\lfloor n/2\rfloor}n\I(\log D)+k^2n\I(\log D)\log D+k^3n\I(\log D))$.

Let us now turn to the second approach. We consider projections of $\Newt(f)$ onto two-dimensional planes. More precisely, for two distinct variables $X_i$ and $X_j$, let us consider $f$ as an element of $\LL_{ij}[X_i,X_j]$ where $\LL_{ij}=\KK[\XX\setminus\{X_i,X_j\}]$ is the ring of polynomials in the other variables. 
If $g\in\PolyRing$ is a unidimensional polynomial of direction $\dd$ with $\delta_i\neq0$, it is still unidimensional when seen as an element of $\LL_{ij}[X_i,X_j]$ for all $j$, and its direction is collinear to $(\delta_i,\delta_j)$. This means that if $g$ divides $f$, the support of $f$ viewed as an element of $\LL_{ij}[X_i,X_j]$ can be partitioned into line segments of direction $(\delta_i,\delta_j)$. Thus, if the support of $f$ can be partitioned into line segments of direction $\dd$,  the support of $f\in\LL_{ij}[X_i,X_j]$ can be partitioned into line segments of direction $(\delta_i,\delta_j)$ for all $i$ and $j$ such that $(\delta_i,\delta_j)\neq (0,0)$. 
Let us define for all $(i,j)$ the set $\Delta_1^{ij}(f)\subset\ZZ^2$ corresponding to $f\in\LL_{ij}[X_i,X_j]$. The set $\Delta_1(f)$ can be computed as follows:
\begin{algorithmic}[1]
\State Compute $\Delta^{ij}(f)$ for $1\le i<j\le n$;
\State $\Delta_1(f)\gets \{\dd\neq 0:\forall i,j, (\delta_i,\delta_j)\in\Delta^{ij}(f)\cup\{(0,0)\}\}$;
\State \textbf{return} $\Delta_1(f)$.
\end{algorithmic}
To analyze the complexity of this algorithm, first note that each $\Delta^{ij}(f)$ can be computed in time $\bigO(k\I(\log D)\log D)$. Even though the size of each $\Delta^{ij}(f)$ can be linear in $k$, the size of $\Delta_1(f)$ is at most quadratic since each pair of monomials of $f$ defines at most one direction $\dd$. Therefore, the total complexity of the algorithm is bounded by $\bigO(k^2n^2\I(\log D)\log D)$. 
\end{proof} 

\subsection{Computing unidimensional factors} 

This section is devoted to an algorithm to compute the unidimensional factors of direction $\dd$ of a lacunary polynomial $f$, as soon as one has an algorithm for factoring lacunary univariate polynomials. One first computes the $\dd$-components of $f$, then their projections, and then the set $\proj\Fcal$ of common factors of these projections, with multiplicities. The set $\Fcal$ of factors of $f$ is then obtained by lifting the elements of $\proj\Fcal$ in direction $\dd$. Next lemma shows that the complexity of this strategy is roughly speaking the complexity of the underlying univariate factorization algorithm.

\begin{lem} \label{lem:unidimcomplexity}
Let $f\in\PolyRing$ be an $n$-variate $k$-nomial of total degree $D$, and $\dd\in\ZZ^n$ a direction with $\|\dd\|_\infty\le D$.
\begin{itemize}\sloppy
\item The $\dd$-components of $f$ can be computed in $\bigO(kn\I(\log D))$ bit-operations;
\item If $f$ is unidimensional of direction $\dd$, its projection can be computed in $\bigO(kn\I(\log D))$ bit-operations;
\item If $g\in\KK[Z]$ has degree $\le D$ and $\ell$ terms, its lifting in direction $\dd$ can be computed in $\bigO(\ell n\I(\log D))$.
\end{itemize}
\end{lem}

\begin{proof} 
The complexity of computing the $\dd$-components is directly given by Lemma~\ref{lem:lines}.
Projection and lifting are computed using arithmetic functions on the components of the vectors, whence the same bound.
\end{proof} 

Altogether, this proves that if one has an algorithm to compute factors of lacunary univariate polynomials, one has an algorithm to computing unidimensional factors of lacunary multivariate polynomials. 
We give here a more formal description of such an algorithm in the case of bounded-degree factors based on Lenstra's algorithm for univariate polynomials~\citep{Len99}.

\begin{thm}
Given an irreducible polynomial $\varphi\in\QQ[\xi]$ in dense representation, a polynomial $f\in\PolyRing$, where $\KK=\QQ[\xi]/\langle\varphi\rangle$, given in lacunary representation, and a multidegree bound $(d_1,\dotsc,d_n)\in\NN^n$, one can compute the unidimensional factors of multidegree at most $(d_1,\dotsc,d_n)$ of $f$ in deterministic time $\poly(\size(f)+\max_id_i)$.
\end{thm}

\begin{proof} 
The algorithm is as follows.
\begin{algorithmic}[1]
\State Compute $\Delta_1(f)$; \Comment{Lemma~\ref{lem:delta1}}
\State $\Fcal\gets\emptyset$;
\ForAll{$\dd\in\Delta_1(f)$}
    \State $\{f^1,\dotsc,f^s\}\gets$ $\dd$-components of $f$; \Comment{Lemma~\ref{lem:unidimcomplexity}}
    \State $d\gets\min_{1\le i\le n} (d_i/|\delta_i|)$;
    \For{$t=1$ to $s$}
        \State $\proj f^t\gets$ projection of $f^t$ in direction $\dd$;
        \State $\proj\Fcal^t\gets$ degree-$d$ factors of $\proj f^t$; \Comment{Lenstra's algorithm}
        \State $\Fcal^t\gets$ the set of liftings in direction $\dd$ of elements of $\proj\Fcal^t$;
    \EndFor
    \State $\Fcal\gets\Fcal\cup \bigcap_t \Fcal^t$;
\EndFor
\State \textbf{return} $\Fcal$.
\end{algorithmic}

The correctness and complexity of this algorithm directly follow from the correctness and complexity of Lenstra's, using the lemmas of this section.
\end{proof} 

\section{Multidimensional factors} \label{sec:multidim} 

In this section, we focus on multidimensional factors. Their computation is based on a Gap Theorem, which follows from a bound on the valuation of an expression $f(X,\phi(X))$ where $f$ is a lacunary polynomial and $\phi\in\puiseux X$ cancels a low-degree polynomial. This bound is given in \refsec{valuation}. The Gap Theorem is stated for bivariate polynomials in \refsec{bivariate} and yields an algorithm for bivariate polynomials which consists in reducing the computation to several factorizations of low-degree polynomials. In \refsec{multivariate}, the bivariate algorithm is first simplified in order to generalize it to multivariate polynomials. 

\subsection{Bound on the valuation} \label{sec:valuation}

The goal of this section is to bound the valuation of an expression $f(X,\phi(X))$ where $f$ is a lacunary polynomial with at least $2$ terms and $\phi$ a Puiseux series with a low-degree minimal polynomial.

To express the bound, let us define 
\[\gamma_v(\ell)=4d_Xd_Y(\ell-1)^2-\frac{1}{2}(\ell-1)((3\ell-4)d_X+v\ell)\]
where the dependency of $\gamma_v(\ell)$ on $d_X$ and $d_Y$ is not explicitly stated since these quantities shall not vary in the following. Note that if $v$ denotes the valuation of a root of a polynomial of bidegree at most $(d_X,d_Y)$, it is bounded in absolute value by $d_X$. Thus, $(3\ell-4)d_X+v\ell\ge (2\ell-4)d_X$. Furthermore, $\ell\ge 2$ shall denote the number of terms of $f$, whence $\gamma_v(\ell)\le 4d_Xd_Y(\ell-1)^2$ for all $v$. We define $\gamma(\ell)=4d_Xd_Y(\ell-1)^2$, so that $\gamma_v(\ell)\le\gamma(\ell)$ for all $|v|\le d_X$.

\begin{thm} \label{thm:val}
Let $\phi\in\puiseux X$ of valuation $v$ and $g\in\KK[X,Y]$ its minimal polynomial of bidegree $(d_X,d_Y)$.  
Let $f=\sum_{j=1}^\ell c_j X^{\alpha_j}Y^{\beta_j}$ be a polynomial with exactly $\ell$ terms,
and suppose that the family $(X^{\alpha_j}\phi(X)^{\beta_j})_{1\le j\le\ell}$ is linearly independent over $\KK$.

Then
\[\val\bigl(f(X,\phi(X))\bigr)\le \min_{1\le j\le\ell}(\alpha_j+v\beta_j) + \gamma_v(\ell)\text.\] 

\end{thm}

The highest order term of the bound in the theorem is $4d_xd_Y\ell^2$. This is provably not tight since for $d_X=d_Y=1$, the bound $\ell^2+\bigO(\ell)$ holds~\citep{ChaGreKoiPoStr13}. 

The proof of Theorem~\ref{thm:val} is based on a series of lemmas. Lemmas~\ref{lemma:composition} to~\ref{lemma:wronskian} appeared in~\citep{KoiPoTa15} in a slightly less precise formulation. Lemma~\ref{lemma:phi-k} can also be found in~\citep{BoChyLeSaSch07}. Lemma~\ref{lemma:polySeries} is pretty classical though it does not explicitly appear as such in the literature. We include proofs for completeness and to obtain better complexity results.

\begin{lem}\label{lemma:composition}
Let $\phi\in\puiseux X$, $h\in\KK[X,Y]$ and $H(X)=h(X,\phi(X))$. For all $u$ and $v$, let $H_{X^uY^v}(X)=(\partial^{u+v} h/\partial X^uY^v)(X,\phi(X))$. Then for all $k\ge 1$,
\[H^{(k)} = \phi^{(k)} H_Y + \sum_{\substack{i_1+\dotsb+i_p\le k\\1\le i_j<k}} c_\ii \phi^{(i_1)}\dotsb\phi^{(i_p)} H_{X^{k-i}Y^p}\]
where the sum ranges over all tuples $\ii=(i_1,\dotsc,i_p)$, $p\ge 0$, with $1\le i_j<k$ for all $j$ and $i=i_1+\dotsc+i_p\le k$. 
\end{lem}

\begin{proof} 
For a tuple $\ii$, let $T^k_\ii=\phi^{(i_1)}\dotsb\phi^{(i_p)} H_{X^{k-i}Y^p}$. We aim to prove by induction on $k$ that $H^{(k)}=\phi^{(k)}H_Y+\sum_\ii c_\ii T^k_\ii$. Note that the empty tuple corresponding to $p=0$ is allowed, in which case we have the term $c_\emptyset H_{X^k}$. 

For $k=1$, the chain rule implies $H'=\phi'H_Y + H_X=0$ and the result holds with $c_\emptyset=1$. 
To compute $H^{(k+1)}$, first note that 
\[(\phi^{(k)}H_Y)'=\phi^{(k+1)} H_Y + \phi^{(k)}(\phi'H_{Y^2}+H_{XY})=\phi^{(k+1)}H_Y + T^{k+1}_{(k,1)}+T^{k+1}_{(k)}\text.\]
Further, the product rule applied to $T^k_\ii$ shows that 
\begin{align*}
(T^k_\ii)' & = \sum_{j=1}^p T^{k+1}_{(i_1,\dotsc,i_j+1,\dotsc,i_p)} + T^{k+1}_{(i_1,\dotsc,i_p,1)}+T^{k+1}_{(i_1,\dotsc,i_p)}\text.
\end{align*}
This proves the lemma since $H^{(k+1)}=(\phi^{(k)}H_Y)'+\sum_\ii c_\ii (T_\ii^k)'$ by induction.
\end{proof} 

\begin{lem}\label{lemma:phi-k}
Let $\phi\in\puiseux X$ with minimal polynomial $g\in\KK[X,Y]$ of bidegree $(d_X,d_Y)$.
Then 
\begin{equation}\label{eq:phi-k}
\phi^{(k)}(X) = \frac{r_k}{g_Y^{2k-1}}(X,\phi(X))
\end{equation}
where $g_Y=\partial g/\partial Y$ and $r_k\in\KK[X,Y]$ satisfies 
\[\begin{cases}
    \deg_X(r_k)\le (2k-1)d_X-k& \text{and}\\
    \deg_Y(r_k)\le (2k-1)d_Y-2(k-1)\text.
\end{cases}\]
\end{lem}

\begin{proof} 
Let $G(X)=g(X,\phi(X))$ and $G_{X^iY^j}(X)=(\partial^{i+j} g/\partial X^iY^j)(X,\phi(X))$ for all nonnegative $i$ and $j$. By definition $G(X)=0$ whence $G^{(k)}(X)=0$ for all $k\ge 0$. By Lemma~\ref{lemma:composition}, $\phi^{(k)}=-\sum_\ii c_\ii T^k_\ii/G_Y$ where $T_\ii^k=\phi^{(i_1)}\dotsb\phi^{(i_p)}G_{X^{k-i}Y^p}$. Let us prove the lemma by induction on $k$. For $k=1$, $\phi'=-G_X/G_Y$ so the lemma holds. Let us assume that the lemma holds for all $j<k$, and consider a term $T^k_\ii$. Let $R_k(X)=r_k(X,\phi(X))$ for all $k$. Then
\begin{align*}
T^k_\ii &= \phi^{(i_1)}\dotsb\phi^{(i_p)} G_{X^{k-i}Y^p}\\
        &= \frac{R_{i_1}}{G_Y^{2i_1-1}} \dotsb\frac{R_{i_p}}{G_Y^{2i_p-1}} G_{X^{k-i}Y^p}\qquad\text{(by induction)}\\
        &= \frac{1}{G_Y^{2k-1}}R_{i_1}\dotsb R_{i_p} G_{X^{k-i}Y^p} G_Y^{2(k-i)+p-2}\text.
\end{align*}
Let $r_{k,\ii}=r_{i_1}\dotsb r_{i_p}g_{X^{k-i}Y^p}g_Y^{2(k-i)+p-2}$ and $r_k=\sum_\ii c_\ii r_{k,\ii}$. To conclude, it is enough to bound the degree of $r_k$. 
By induction, $\deg_X(r_i)\le(2i-1)d_X-i$ and $\deg_Y(r_i)\le(2i-1)d_Y-2(i-1)$. Whence
\begin{align*}
\deg_X(r_{k,\ii})&\le \sum_{j=1}^p ((2i_j-1)d_X-i_j)+ (d_X-k+i)+(2(k-i)+p-2) d_X \\
                & \le (2i-p)d_X-i+d_X-(k-i)+(2(k-i)+p-2)d_X\\
                & \le(2k-1)d_X-k
\end{align*}
and $\deg_Y(r_{k,\ii})\le (2k-1)d_X-2(k-1)$ with a similar computation. 
\end{proof} 

\begin{lem}\label{lemma:psi-k}
Let $\phi\in\puiseux X$ with minimal polynomial $g\in\KK[X,Y]$ of bidegree $(d_X,d_Y)$. Let $\psi(X)=X^\alpha\phi(X)^\beta$ for some integer $\alpha,\beta\ge k$. Then
\[\psi^{(k)}(X)=X^{\alpha-k}\phi(X)^{\beta-k}\frac{s_k}{g_Y^{2k-1}}(X,\phi(X))\]
where $g_Y=\partial g/\partial Y$ and $s_k$ satisfies
\[\begin{cases}
\deg_X(s_k)\le (2k-1)d_X&\text{and}\\
\deg_Y(s_k)\le (2k-1)d_Y-(k-1)\text.
\end{cases}\]
\end{lem}

\begin{proof} 
By Lemma~\ref{lemma:composition} with $h(X,Y)=X^{\alpha}Y^{\beta}$, 
\[\psi^{(k)}=\phi^{(k)} X^{\alpha}\beta\phi^{\beta-1}+\sum_{\ii} \tilde c_\ii \phi^{(i_1)}\dotsb\phi^{(i_p)} X^{\alpha-k+i}\phi^{\beta-p}\]
where $\tilde c_\ii=(\alpha)_{k-i}(\beta)_pc_\ii$ for all $\ii$. 
As previously, let $G_Y(X)=g_Y(X,\phi(X))$ and $R_k(X)=r_k(X,\phi(X))$ for all $k$. By Lemma~\ref{lemma:phi-k},
\[\phi^{(k)}X^\alpha\beta\phi^{\beta-1} = \frac{X^{\alpha-k}\phi^{\beta-k}}{G_Y^{2k-1}} \left(\beta R_k X^k\phi^{k-1}\right)\] 
and 
\begin{align*}
\phi^{(i_1)}\dotsb\phi^{(i_p)} X^{\alpha-k+i}\phi^{\beta-p} 
&= \frac{1}{G_Y^{2i-p}} R_{i_1}\dotsb R_{i_p} X^{\alpha-k+i}\phi^{\beta-p}\\
&= \frac{X^{\alpha-k}\phi^{\beta-k}}{G_Y^{2k-1}} \left(R_{i_1}\dotsb R_{i_p}X^i\phi^{k-p} G_Y^{2(k-i)+p-1}\right)\text.
\end{align*}
The function 
\[s_k=\beta X^kY^{k-1}r_k+\sum_\ii \tilde c_\ii X^iY^{k-p} r_{i_1}\dotsb r_{i_p}g_Y^{2(k-i)+p-1}\]
satisfies the lemma and it only remains to bound its degree.

To this end, a simple computation shows that 
\begin{align*}
\deg_X(s_k)&\le \max_\ii\bigl(i+(2i-p)d_X-i+(2(k-i)+p-1)d_X\bigr)\\
           & \le (2k-1)d_X
\end{align*}
and similarly $\deg_Y(s_k)\le (2k-1)d_Y-(k-1)$.
\end{proof} 

\begin{lem}\label{lemma:wronskian}
Let $\phi\in\puiseux X$ with minimal polynomial $g\in\KK[X,Y]$ of bidegree $(d_X,d_Y)$. For $1\le j\le\ell$, let $\psi_j(X)=X^{\alpha_j}\phi(X)^{\beta_j}$ where $\alpha_j$ and $\beta_j$ are nonnegative integers. If the family $(\psi_j)_j$ is linearly independent, then
\begin{equation}\label{eq:wronskian}
\wr(\psi_1,\dotsc,\psi_\ell)=X^{A-\binom{\ell}{2}}\phi(X)^{B-\binom{\ell}{2}} \frac{t_\ell(X,\phi(X))}{g_Y^{(\ell-1)^2}(X,\phi(X))}
\end{equation}
where $A=\sum_j\alpha_j$, $B=\sum_j\beta_j$ and $t_\ell\in\KK[X,Y]$ satisfies 
\[\begin{cases} 
    \deg_X(t_\ell)\le (\ell-1)^2d_X &\text{and}\\
    \deg_Y(t_\ell)\le (\ell-1)^2d_Y-\binom{\ell-1}{2}\text.
\end{cases}\]
\end{lem}

\begin{proof} 
Let us first assume that $\alpha_j,\beta_j\ge\ell$ and express the wronskian using the full symbolic expansion of the determinant. 
It is a sum of terms of the form $\psi_1^{(k_1)}\dotsb \psi_\ell^{(k_\ell)}$
such that $\{k_1,\dotsc,k_\ell\}=\{0,\dotsc,\ell-1\}$.
By Lemma~\ref{lemma:psi-k}, 
\[\psi_j^{(k-1)}= X^{\alpha_j-k+1}\phi(X)^{\beta_j-k+1} s_{k-1}(X,\phi(X))/g_Y^{2k-3}(X,\phi(X))\]
for $k\ge 2$. (Note that in the above expression, $s_{k-1}$ actually depends on $j$.) Thus each term in the wronskian has the form
\[X^{\sum_j\alpha_j-\binom{\ell}{2}}\phi(X)^{\sum_j\beta_j-\binom{\ell}{2}}\frac{\tilde t_\ell}{g_Y^{(l-1)^2}}(X,\phi(X)) \]
where $\tilde t_\ell=\prod_{k\ge 2}s_{k-1}$. Thus
\[\deg_X(\tilde t_\ell)\le \sum_{k=1}^{\ell-1} (2k-1)d_X=(\ell-1)^2d_X\]
and $\deg_Y(\tilde t_\ell)\le (\ell-1)^2d_Y-\binom{\ell-1}{2}$. 
To conclude, $t_\ell$ is defined as the sum of the $\tilde t_\ell$'s and the degree bounds still hold.

To remove the assumption $\alpha_j,\beta_j\ge\ell$, we apply the above proof to the family $(\chi_j)_j$ defined by 
$\chi_j=X^\ell\phi^\ell\psi_j$, and that satisfies $\wr(\chi_1,\dotsc,\chi_\ell)=X^{\ell^2}\phi(X)^{\ell^2}\wr(\psi_1,\dotsc,\psi_\ell)$.
\end{proof} 

\begin{lem}\label{lemma:polySeries}
Let $\phi\in\puiseux X$ with minimal polynomial $g\in\KK[X,Y]$ of bidegree $(d_X,d_Y)$. Let $h(X,Y)$ be a polynomial of bidegree $(\delta_X,\delta_Y)$. Then 
\[\left|\val(h(X,\phi(X)))\right|\le d_X\delta_Y+\delta_X d_Y\text.\]
\end{lem}

\begin{proof}  
The resultant 
$r(X,Y)=\res_Z(g(X,Z),Y-h(X,Z))$
vanishes for $Y=h(X,\phi(X))$ since both $g(X,Z)$ and $h(X,\phi(X))-h(X,Z)$ vanish when $Z=\phi(X)$.
In the Sylvester matrix associated to this resultant, $\delta_Y$ rows are made of the coefficients of $g$ viewed as a polynomial in $Z$, and $d_Y$ rows are made of the coefficients of $Y-h(X,Z)$. Since $\deg_X(g) = d_X$ and $\deg_X(Y-h(X,Z))=\delta_X$, each term of the resultant has degree at most $d_X\delta_Y+\delta_Xd_Y$ in $X$. 

We have shown that $h(X,\phi)$ is a Puiseux series which cancels a polynomial $r$ of degree at most $(d_X\delta_Y+\delta_X d_Y)$ in $X$. By Newton-Puiseux Theorem, this quantity also bounds the absolute value of its valuation.
\end{proof}  

\begin{proof}[Proof of Theorem~\ref{thm:val}] 
Let $W=\wr(X^{\alpha_1}\phi^{\beta_1}, \dotsc, X^{\alpha_\ell}\phi^{\beta_\ell})$ and $F(X)=f(X,\phi(X))$. Without loss of generality, let us assume that $\min_j(\alpha_j+v\beta_j)$ is attained for $j=1$. 

Since $F$ is a linear combination of the family $(X^{\alpha_j}\phi(X)^{\beta_j})_j$, the wronskian $W_F$ of the family $(F, X^{\alpha_2}\phi^{\beta_2}, \dotsc, X^{\alpha_\ell}\phi^{\beta_\ell})$ satisfies $W_F=a_1W$ and their valuations coincide. By Lemma~\ref{lem:valLowerBound}, 
\[\val(W_F)\ge\val(F)+\sum_{j>1} (\alpha_j+v\beta_j)-\binom{\ell}{2}\text.\]

On the other hand, as the family $(X^{\alpha_j}\phi^{\beta_j})_j$ is linearly independent, Lemma~\ref{lemma:wronskian} implies the existence of a nonzero $t_\ell$ such that
\[W=X^{A-\binom{\ell}{2}}\phi^{B-\binom{\ell}{2}}\frac{t_\ell(X,\phi)}{g_Y^{(\ell-1)^2}(X,\phi)}\text.\]
Moreover, Lemma~\ref{lemma:polySeries} implies that $\val(t_\ell(X,\phi))\le 2d_Xd_Y(\ell-1)^2-d_X\binom{\ell-1}{2}$ 
and $\val(g_Y(X,\phi))\ge -2d_Xd_Y+d_X$. 
Therefore, 
\[\val(W)\le A-\binom{\ell}{2}+vB-v\binom{\ell}{2}+\frac{1}{2}d_X(\ell-1)(8d_Y(\ell-1)-3\ell+4)\text.\]

Since $A=\sum_j\alpha_j$ and $B=\sum_j\beta_j$, 
\[\val(F)\le \alpha_1+v\beta_1 -v\binom{\ell}{2}+\frac{1}{2}d_X(\ell-1)(8d_Y(\ell-1)-3\ell+4)\text.\] 
The theorem follows, since $\alpha_1+v\beta_1=\min_j(\alpha_j+v\beta_j)$.
\end{proof} 

\subsection{The bivariate case} \label{sec:bivariate} 

In this section we state a Gap Theorem for bivariate polynomials which derives from Theorem~\ref{thm:val}, and deduce an algorithm to compute the multidimensional factors of bounded bidegree of a lacunary polynomial.

In order to simplify the exposition, we shall use the bound $\gamma_v(\ell)\le\gamma(\ell)$, valid as soon as $|v|\le d_X$. Using $\gamma_v(\ell)$ instead of $\gamma(\ell)$ yields slightly better results at the price of much more cumbersome proofs.

\begin{thm}[Gap Theorem]
Let $v\in\QQ$, $d_X$, $d_Y\in\NN$ and $f=f_1+f_2\in\KK[X,Y]$ where
\[f_1 = \sum_{j=1}^\ell c_j X^{\alpha_j}Y^{\beta_j}\quad\text{ and }\quad f_2=\sum_{j=\ell+1}^k c_j X^{\alpha_j} Y^{\beta_j}\]
satisfy $\alpha_j+v\beta_j\le\alpha_{j+1}+v\beta_{j+1}$ for $1\le j<k$. If $\ell$ is the smallest index such that
\[\alpha_{\ell+1}+v\beta_{\ell+1} > \alpha_1+v\beta_1 + \gamma(\ell)\text,\]
then for every irreducible polynomial of bidegree at most $(d_X,d_Y)$ such that $g$ has a root of valuation $v$ in $\puiseux X$,
\[\mult_g(f)=\min(\mult_g(f_1),\mult_g(f_2))\text.\]
\end{thm}

\begin{proof} 
Let us first prove that under the assumptions of the theorem, $g$ divides $f$ if and only if it divides both $f_1$ and $f_2$. 

Consider a polynomial $g$ as in the theorem, and $\phi\in\puiseux X$ such that $g(X,\phi(X))=0$. Since $g$ is irreducible, it divides $f$ if, and only if, $f(X,\phi(X))=0$. Let us assume that $g$ does not divide $f_1$, that is, $f_1(X,\phi(X))\neq 0$, and consider the family $(X^{\alpha_j}\phi^{\beta_j})_{1\le j\le\ell}$. One can extract a basis $(X^{\alpha_{j_t}}\phi^{\beta_{j_t}})_{1\le t\le m}$ of this family and rewrite 
\[f_1(X,\phi(X))=\sum_{t=1}^m b_t X^{\alpha_{j_t}}\phi(X)^{\beta_{j_t}}\]
where $b_1$, \dots, $b_m$ are linear combinations of $c_1$, \dots, $c_\ell$. Without any loss of generality, let us assume that $b_t\neq0$ for all $t$. Since $(X^{\alpha_{j_t}}\phi^{\beta_{j_t}})_{1\le t\le m}$ is linearly independent, Theorem~\ref{thm:val} implies 
\[\val(f_1(X,\phi(X)))\le \min_t (\alpha_{j_t}+v\beta_{j_t})+\gamma(m)\text.\]
By minimality of $\ell$, $\alpha_{j_t}+v\beta_{j_t}\le \alpha_1+v\beta_1+\gamma(j_t-1)$ for all $t$. Since $j_t+m-1\le \ell$ and $\gamma(\ell_1)+\gamma(\ell_2) \le\gamma(\ell_1+\ell_2)$ for all $\ell_1$ and $\ell_2$,
\[\val(f_1(X,\phi(X)))\le \alpha_1+v\beta_1+\gamma(\ell)\text.\]
By assumption $\val(f_2(X,\phi(X)))>\val(f_1(X,\phi(X)))$, whence $f_1(X,\phi(X))+f_2(X,\phi(X))\neq0$. In other words, if $g$ does not divide $f_1$, it does not divide $f$ either.

To obtain the statement on the multiplicities, consider the $p$-th derivatives of $f$, $f_1$ and $f_2$. Then $\mult_g(f)>p$ if and only if $g$ divides $f^{(p)}$. Let us assume without loss of generality that $\alpha_j,\beta_j>p$ for all $j$. (For one can multiply $f$ by $X^pY^p$ without changing its irreducible factors but the multiplicity of $X$ and $Y$ as factors of $f$.) One can write $f^{(p)} = f_1^{(p)}+f_2^{(p)}$. Furthermore, the condition of the lemma is satisfied by $f^{(p)}$ if and only if it is satisfied by $f$ since it is based on the difference of the exponents. Thus, for all $p$, $g$ divides $f^{(p)}$ if and only if it divides both $f_1^{(p)}$ and $f_2^{(p)}$. The conclusion follows.
\end{proof} 

In order to avoid any misunderstanding, we formalize what it means for a polynomial $f$ to have a gap relative to a valuation $v$.

\begin{defn}
Let $v\in\QQ$, $d_X$, $d_Y\in\NN$ and $f = \sum_{j=1}^k c_j X^{\alpha_j}Y^{\beta_j}$ such that $\alpha_j+v\beta_j\le\alpha_{j+1}+v\beta_{j+1}$ for $1\le j<k$.
We say that \emph{$f$ has no gap relative to $v$} if for $1<\ell\le k$,
\[\alpha_\ell+v\beta_\ell \le \alpha_1+v\beta_1 + \gamma(\ell-1)\text.\]
Otherwise, \emph{$f$ has a gap relative to $v$}.
\end{defn}

The Gap Theorem can be used to partition an input polynomial $f$ into a sum $f_1+\dotsb+f_s$ such that for all irreducible polynomial $g$ of bidegree at most $(d_X,d_Y)$ with a root of valuation $v$, $\mult_g(f)=\min_t (\mult_g(f_t))$. Graphically, this partition corresponds to a partition of the support of $f$ into oblique strips, each of which has width $\gamma(\ell_t)$ where $\ell_t$ is the number of points of the support of $f$ it contains. This is the algorithm \textsc{Partition} (Algorithm~\ref{algo:partition}).

\begin{algorithm}[htp]
\caption{$\protect\Call{Partition}{f,d_X,d_Y,v,\sigma}$}
\label{algo:partition}
\begin{algorithmic}[1]
\Input $f=\sum_{j=1}^kc_j X^{\alpha_j}Y^{\beta_j}$, $d_X$, $d_Y\in\NN$, $v\in\QQ$ and $\sigma$ a permutation such that $j\mapsto\alpha_{\sigma(j)}+v\beta_{\sigma(j)}$ is non-decreasing;
\Output A partition $f=f_1+\dotsb+f_s$.
\Statex
\State $t\gets 1$;
\State $j_m\gets 1$;
\For{$\ell=2$ to $k$}
    \If{$\alpha_{\sigma(\ell)}+v\beta_{\sigma(\ell)}>\alpha_{\sigma(j_m)}+v\beta_{\sigma(j_m)}+\gamma(\ell-j_m)$}
        \State $f_t\gets \sum_{j=j_m}^{\ell-1} c_{\sigma(j)} X^{\alpha_{\sigma(j)}} Y^{\beta_{\sigma(j)}}$;
        \State $j_m\gets \ell+1$;
        \State $t\gets t+1$;
    \EndIf
\EndFor
\State $f_t\gets \sum_{j=j_m}^{k} c_{\sigma(j)} X^{\alpha_{\sigma(j)}} Y^{\beta_{\sigma(j)}}$;
\State \textbf{return} $\{f_1,\dotsc,f_t\}$.
\end{algorithmic}
\end{algorithm}

\begin{lem}\label{lem:partition}
If $f$ is a bivariate $k$-nomial of total degree $D$, and $d_X,d_Y\le D$, the algorithm $\Call{Partition}{f,d_X,d_Y,v,\sigma}$ runs in time $\bigO(k\log D)$ and outputs a partition $f_1+\dotsb+f_s$ of $f$ 
such that $\mult_g(f)=\min_{1\le t\le s} (\mult_g(f_t))$ for all irreducible polynomials $g$ such that the lower hull of $\Newt(g)$ contains an edge of direction $(p,q)$ with $v=-p/q$. 
\end{lem}

\begin{proof} 
The correctness of the algorithm is a direct consequence of the Gap Theorem: Indeed, $g$ has a root of valuation $v=-p/q$ in this case. 
The complexity of the algorithm is bounded by $O(k\log D)$ since there are only comparisons of integers of size at most $\log D$. 
\end{proof} 

Alone, this partition does not bound the degree in $X$ nor the degree in $Y$ of each $f_t$. Using the graphical interpretation, the support of $f$ is partitioned into strips that have finite width but are infinite though. The idea is then to use another set of strips, not parallel to the first ones, to refine the partition. Since the strips are not parallel, this will partition the support into parallelograms which are finite.

There comes multidimensionality. If $g$ is multidimensional, $\Newt(g)$ has by definition two non-parallel edges. Let us assume that these two edges belong to the lower hull of the Newton polygon. Then $g$ has a root of valuation $v_1$ and another one of valuation $v_2\neq v_1$ where $v_1$ and $v_2$ are determined by the directions of the two non-parallel edges. One can partition $f$ with respect to $v_1$ and then each summand in the partition can be again partitioned, this time with respect to $v_2$. This yields the algorithm \textsc{Bipartition} (Algorithm~\ref{algo:bipartition}). 

\begin{algorithm}[htp]
\caption{$\protect\Call{Bipartition}{f,d_X,d_Y,v_1,v_2}$}
\label{algo:bipartition}
\begin{algorithmic}[1]
\Input $f=\sum_{j=1}^k c_j X^{\alpha_j}Y^{\beta_j}$, $d_X$, $d_Y\in\NN$, $v_1$, $v_2\in\QQ$;
\Output A partition $f=f_1+\dotsb+f_s$.
\Statex
\For{$i=1,2$}
    \State $\sigma_i\gets$ permutation such that $j\mapsto \alpha_{\sigma_i(j)}+v_i\beta_{\sigma_i(j)}$ is non-decreasing;

\EndFor
\State $\Scal\gets \{f\}$ and $s\gets0$;
\While{$|\Scal|>s$}
    \State $s\gets |\Scal|$;
    \ForAll{$h\in\Scal$}
        \State $\Scal_h\gets\Call{Partition}{h,d_X,d_Y,v_1,\sigma_1}$;
        \ForAll{$h'\in\Scal_h$}
            \State $\Scal_{h'}\gets\Call{Partition}{h',d_X,d_Y,v_2,\sigma_2}$; \label{line:secondpartition}
        \EndFor
        \State $\Scal_h\gets \bigcup_{h'\in\Scal_h} \Scal_{h'}$;
    \EndFor
    \State $\Scal \gets\bigcup_{h\in\Scal} \Scal_h$;
\EndWhile
\State \textbf{return} $\Scal$.
\end{algorithmic}
\end{algorithm}

\begin{lem}\label{lem:bipartition}
If $f$ is a bivariate $k$-nomial of total degree $D$, and $d_X,d_Y\le D$, the algorithm $\Call{Bipartition}{f,d_X,d_Y,v_1,v_2}$ runs in time $\bigO(k^2\log D)$ and outputs a partition $f_1+\dotsb+f_s$ 
such that $\mult_g(f)=\min_{1\le t\le s} (\mult_g(f_t))$ for all multidimensional polynomial $g$ such that the lower hull of $\Newt(g)$ contains two edges of directions $(p_1,q_1)$ and $(p_2,q_2)$ with $v_1=-p_1/q_1$ and $v_2=-p_2/q_2$. 

Furthermore, the convex size of each $f_t$ is at most $\bigO(d_X^2d_Y^2\ell_t^4)$ where $\ell_t$ is the number of terms of $f_t$.
\end{lem}

\begin{proof} 
Again, the correctness of this algorithm directly follows from the correctness of the algorithm of Lemma~\ref{lem:partition}, that is ultimately from the Gap Theorem. To estimate its complexity, first note that the total number of monomials in $\Scal$ remains constant, equal to $k$, during the computation. At each iteration of the while loop, the procedure \textsc{Partition} is called on polynomials whose total number of monomials is $2k$. Thus the complexity is $O(k^2\log D)$ since the sorting phase can be also performed within this complexity bound.

To bound the convex sizes, note that there is no gap anymore in any $f_t$ at the end of the algorithm. The support of each $f_t$ is therefore contained in two strips of widths bounded by $\gamma(\ell_t)$, that is in a parallelogram of area $\gamma(\ell_t)^2=16d_X^2d_Y^2(\ell_t-1)^4$.
\end{proof} 

From the previous lemma, we obtain a reduction to low-degree factorization for bivariate polynomials. Note that we do not state any degree bound in the next theorem (such bounds are given in the next section) but rather a bound on the convex size of the output polynomials. To really have an algorithm to compute bounded-degree factors of bivariate polynomials, one can branch any bivariate factorization algorithm. In order to get the best complexity bounds, one can preprocess the output polynomials before their factorization with the techniques of \citet*{BeLe12}. This allows to compute the irreducible factorization of a polynomial in time polynomial in the convex size rather than the degree of the polynomial. This is particularly interesting in our settings when the support of the input polynomial is partitioned into very \emph{flat} parallelograms, that is parallelograms of large dimensions but small area. 

\begin{thm} \label{thm:bivariate}
Let $f\in\KK[X,Y]$ be a $k$-nomial of total degree $D$ and $d_X$, $d_Y\in\NN$ be degree bounds. One can reduce the computation of the multidimensional bidegree-$(d_X,d_Y)$ factors of $f$ to the irreducible factorization of at most $k$ polynomials of convex size $\bigO(d_X^2d_Y^2k^4)$ in time $\poly(k,\log D,d_X,d_Y)$.
\end{thm}

\begin{proof} 
Note first that the Newton polygon of a multidimensional factor $g$ has two non-parallel edges. There are three possible cases: Either the lower hull of $\Newt(g)$ has two non parallel-edges, or its upper hull has two parallel edges, or $\Newt(g)$ has two vertical edges. These three cases are treated separately.

Let us first only consider factors $g$ such that the lower hull of $\Newt(g)$ has two non-parallel edges. By Corollary~\ref{cor:ostrowski}, the respective directions of these edges must also be directions of edges of the lower hull of $\Newt(f)$. This yields the following algorithm, where $\Call{Gcd}{\cdot}$ denotes a procedure to compute the gcd of a set of polynomials:
\begin{algorithmic}[1]
\Statex
\State $\Scal\gets\emptyset$;
\State $\Delta\gets$ the directions of the edges in the lower hull of $\Newt(f)$;
\State $\Delta\gets \Delta\cap\{(p,q): p\le d_X, |q|\le d_Y\}$;
\ForAll{$(p_1,q_1),(p_2,q_2)\in\Delta$}
    \State $\Scal_{p_1,q_1,p_2,q_2}=\Call{Bipartition}{f,d_X,d_Y,-p_1/q_1,-p_2/q_2}$\;
    \State $h \gets \Call{Gcd}{\Scal_{p_1,q_1,p_2,q_2}}$;
    \State $\Scal \gets \Scal \cup h$;
\EndFor
\State \textbf{return} $\Scal$.
\Statex
\end{algorithmic}

The correctness of this algorithm is a direct result of the correctness of \textsc{Bipartition}.
As for the complexity, there are at most $\bigO(k)$ edges in $\Delta$, whence at most $\bigO(k^2)$ iterations of the loop. The call to \textsc{Bipartition} takes polynomial time. Now, using the techniques of \citet*{BeLe12}, one can compute the gcd of $\Scal$ in time polynomial in the convex size of the elements of $\Scal$. This convex size is bounded by $\bigO(d_X^2d_Y^2k^4)$ according to Lemma~\ref{lem:bipartition}. 

It remains to prove that one can give similar algorithms for the factors of $f$ that have two non-parallel edges in the upper hull of the Newton polygon, or two vertical edges. For the first remaining case, one can simply consider $f^X(X,Y)=Y^{\deg_X(f)}f(1/X,Y)$ and apply the previous algorithm to $f^X$. 
The last case is a bit more different, though the algorithm is actually slightly simpler. One has to slightly modify the algorithm \textsc{Bipartition}. Since there is only one valuation to call \textsc{Partition}, one has to replace the second call to \textsc{Partition} at line~\ref{line:secondpartition} of \textsc{Bipartition}. For, let us define $\bar f(X,Y)=f(Y,X)$ and $\bar g(X,Y)=g(Y,X)$ for any $g$. Clearly, $\mult_{\bar g}(\bar f)=\mult_g(f)$ and if $g$ has two vertical edges, $\bar g$ has two horizontal edges. This means that one can replace the call $\Call{Partition}{h',d_X,d_Y,v_2,\sigma_2}$ at line~\ref{line:secondpartition} by $\Call{Partition}{\bar h',d_X,d_Y,0,\sigma_0}$ where $\sigma_0$ is the permutation such that $\beta_{\sigma_0(j)}\le\beta_{\sigma_0(j+1)}$ for all $j$. 
The rest of the algorithm is identical.
\end{proof} 

\subsection{The multivariate case} \label{sec:multivariate} 

In this section, we aim to generalize Theorem~\ref{thm:bivariate} to multivariate polynomials. Actually, the algorithm is a simplification of the previous one that can be used for bivariate polynomials as well. Yet the price for the simplicity is an increase of computational complexity.

The first step is a new analysis of the algorithm \textsc{Bipartition}. We have given a bound on the convex size of the polynomials in the output. For our simplified algorithm, we need a bound on their degree. 

\begin{lem}\label{lem:boundsval}
Let $(p_1,q_1)$ and $(p_2,q_2)$ be the directions of two non-parallel edges in the lower hull of the Newton polygon of some polynomial of bidegree $(d_X,d_Y)$, and let $v_1=p_1/q_1$ and $v_2=p_2/q_2$. 
Then 
\[ \frac{1}{|v_1-v_2|} \le d_Y^2/4 \text{ and }  \frac{|v_1|+|v_2|}{|v_1-v_2|}\le d_Xd_Y\text.\]
\end{lem}

\begin{proof} 
Since $(p_1,q_1)$ and $(p_2,q_2)$ are the directions of two edges in the lower hull of some Newton polygon, there exist $\lambda$, $\mu\in\NN$ such that $\lambda|q_1|+\mu|q_2|\le d_Y$, whence $|q_1|+|q_2|\le d_Y$. And for similar reasons, $|p_1|,|p_2|\le d_X$. 

Since $|q_1|+|q_2|\le d_Y$, $|q_1q_2|\le |q_1|(d_Y-|q_1|)\le (d_Y/2)^2$. Thus 
\[\frac{1}{|v_1-v_2|} = \frac{|q_1q_2|}{|p_1q_2-p_2q_1|}\le |q_1q_2|\le \frac{d_Y^2}{4}\text.\]

Similarly, $|p_1q_2|+|p_2q_1|\le d_X(|q_1|+(d_Y-|q_1|))\le d_Xd_Y$ and
\[\frac{|v_1|+|v_2|}{|v_1-v_2|} = \frac{|p_1q_2|+|p_2q_1|}{|p_1q_2-p_2q_1|}\le d_Xd_Y\text.\qedhere\]
\end{proof} 

\begin{lem}\label{lem:degree-bounds}
Let $\Scal=\Call{Bipartition}{f,d_X,d_Y,v_1,v_2}$
where $d_X$, $d_Y\in\NN$, $v_1$, $v_2\in\QQ$ and $f\in\KK[X,Y]$. 
Then for all $h\in\Scal$ with $\ell$ terms, 
\[\begin{cases}
\deg_X(h)-\val_X(h)\le \frac{1}{2} d_Y^2\gamma(\ell)&\text{and}\\
\deg_Y(h)-\deg_Y(h)\le d_Xd_Y\gamma(\ell)\text.
\end{cases}\]
\end{lem}

\begin{proof} 
Let $h=\sum_{j=1}^\ell c_jX^{\alpha_j}Y^{\beta_j}\in\Scal$. Since there is no gap in $h$, for all $j$,
\[\begin{cases}
    \alpha_j+v_1\beta_j\le \min_{1\le j\le\ell} (\alpha_j+v_1\beta_j) + \gamma(\ell)&\text{and}\\
    \alpha_j+v_2\beta_j\le \min_{1\le j\le\ell} (\alpha_j+v_2\beta_j) + \gamma(\ell)\text.
\end{cases}\] 
In particular, for all $p$ and $q$, and $i=1,2$,
$(\alpha_p-\alpha_q)+v_i(\beta_p-\beta_q)\le \gamma(\ell)$.
Let us fix some $p$ and some $q$, and let $\Delta_\alpha=\alpha_p-\alpha_q$ and $\Delta_\beta=\beta_p-\beta_q$.
We aim to bound $|\Delta_\alpha|$ and $|\Delta_\beta|$.

Since the above bound is valid if we exchange $p$ and $q$, we have $|\Delta_\alpha+v_i\Delta_\beta|\le\gamma(\ell)$ for $i=1,2$. 
Hence, $\Delta_\alpha+v_1\Delta_\beta-(\Delta_\alpha+v_2\Delta_\beta)=(v_1-v_2)\Delta_\beta\le 2\gamma(\ell)$ and by Lemma~\ref{lem:boundsval}, 
\[|\Delta_\beta| \le \frac{2\gamma(\ell)}{|v_1-v_2|}\le \frac{1}{2} d_Y^2\gamma(\ell)\text.\]

Furthermore, $v_2\Delta_\alpha+v_1v_2\Delta_\beta\le|v_2|\gamma(\ell)$ and $v_1\Delta_\alpha+v_1v_2\Delta_\beta\ge -|v_1|\gamma(\ell)$. Thus $(v_2-v_1)\Delta_\alpha\le(|v_1|+|v_2|)\gamma(\ell)$ and 
\[ |\Delta_\alpha| \le \frac{|v_1|+|v_2|}{|v_1-v_2|}\gamma(\ell) \le d_Xd_Y\gamma(\ell)\]
again using Lemma~\ref{lem:boundsval}. This proves the lemma.
\end{proof} 

Until now, we have obtained for each pair of distinct valuations $(v_1,v_2)$ a partition of $f$ with the desired properties. We aim to invert the quantifiers, that is to prove that there exists a partition that has the desired properties with respect to any pair of valuations. 

\begin{lem}\label{lem:univariatepartition}
Let $f\in\KK[X,Y]$ be a $k$-nomial of total degree $D$, and $d_X$, $d_Y\in\NN$. There exists a partition $f=f_1+\dotsb+f_s$ such that for any multidimensional polynomial $g$ of bidegree at most $(d_X,d_Y)$, 
\[\mult_g(f)=\min_{1\le t\le s} (\mult_g(f_t))\text.\]
Furthermore, for $1\le t\le s$
\[\deg_Y(f_t)-\val_Y(f_t)\le kd_Xd_Y\gamma(k)\text.\]
This partition can be computed in time $\bigO(k\log D)$.
\end{lem}

\begin{proof} 
Let us write $f=\sum_{j=1}^k c_j X^{\alpha_j} Y^{\beta_j}$ such that $\beta_j\le\beta_{j+1}$ for $1\le j<k$.  
Consider the following algorithm that computes a partition of $f$:
\begin{algorithmic}[1]
\State $t\gets 1$;
\State $j_m \gets 1$;
\For{$\ell=1$ to $k-1$}
    \If{$\beta_{\ell+1}-\beta_{\ell}> d_Xd_Y\gamma(k)$} 
        \State $f_t\gets\sum_{j=j_m}^\ell c_j X^{\alpha_j} Y^{\beta_j}$;
        \State $j_m\gets\ell+1$;
        \State $t\gets t+1$;
    \EndIf
\EndFor
\State \textbf{return} $\{f_1,\dotsc,f_t\}$.
\end{algorithmic}
Let us first note that the bound on $\deg_Y(f_t)-\val_Y(f_t)$ is straightforward since the number of terms of $f_t$ is bounded by $k$. The complexity is dominated by the cost of sorting the exponents, and this cost is $\bigO(k\log D)$. 

Let $s$ be the final value of $t$ in the above algorithm. We claim that at the end of the algorithm, the partition $f=f_1+\dotsb+f_s$ satisfies $\mult_g(f)=\min_t (\mult_g(f_t))$ for all multidimensional polynomial $g$ of bidegree at most $(d_X,d_Y)$. For, let us fix such a polynomial $g$. 

First, if the lower hull of $\Newt(g)$ has two non-parallel edges, $g$ has two roots of distinct valuations $v_1$ and $v_2$ in $\puiseux X$. The call $\Call{Bipartition}{f,d_X,d_Y,v_1,v_2}$ computes a partition $f=f_1^1+\dotsb+f_{s^1}^1$. This partition satisfies $\deg_Y(f_t^1)-\val_Y(f_t^1)\le d_Xd_Y\gamma(\ell_t^1)$ where $\ell_t^1$ is the number of terms of $f_t^1$ and $\mult_g(f)=\min_t (\mult_g(f_t^1))$. We aim to show that if two monomials $X^{\alpha_p}Y^{\beta_p}$ and $X^{\alpha_q}Y^{\beta_q}$ belong to a same $f_t^1$ in this partition, they also belong to a same polynomial $f_t$ in the partition computed by the algorithm. Indeed, given the bound on $\deg_Y(f_t^1)-\val_Y(f_t^1)$, we have $|\alpha_p-\alpha_q|\le d_Xd_Y\gamma(\ell_t^1)\le d_Xd_Y\gamma(k)$ since $\gamma$ is an increasing function. Let us assume $\alpha_p\ge\alpha_q$. Then, for $q< r\le p$, $\alpha_r-\alpha_{r-1}\le\alpha_p-\alpha_q\le d_Xd_Y\gamma(k)$ and $\alpha_p$, $\alpha_{p+1}$, \dots, $\alpha_q$ belong to a same polynomial in the partition $f=f_1+\dotsb+f_t$. This proves that each $f_t$ is a sum of $f_{t'}^1$'s, hence $\min_t(\mult_g(f_t))=\min_t(\mult_g(f_t^1))$. 

The two other cases concern polynomials such that the upper hull of their Newton polygon has two non-parallel edges, and polynomials with two vertical edges. As in the proof of Theorem~\ref{thm:bivariate}, one can consider $f^X(X,Y)= X^{\deg_X(f)}f(1/X,Y)$ for the first case, and $\bar f(X,Y) = f(Y,X)$ for the second case to complete the proof.
\end{proof} 

Finally, we get to our simple algorithm. First note that we can replace the bound $\gamma(k)d_Xd_Y$ in the algorithm by any larger value and obtain the same result, but of course with a larger value for $\deg_Y(f_t)-\val_Y(f_t)$. Our aim is to use the above partitioning algorithm with respect to all the variables, sequentially. For ease of presentation, let us reformulate the above algorithm in the settings of a multivariate polynomial (\textsc{UnivariatePartition}, Algorithm~\ref{algo:univariatepartition}) before presenting the general algorithm (\textsc{MultivariatePartition}, Algorithm~\ref{algo:multivariatepartition}).

\begin{algorithm}[htbp]
\caption{$\protect\Call{UnivariatePartition}{f,\delta,i}$}
\label{algo:univariatepartition}
\begin{algorithmic}[1]
\Input $f=\sum_{j=1}^k c_j X_1^{\alpha_{1,j}}\dotsb X_n^{\alpha_{n,j}}$, $\delta$ and $i$;
\Output $\{f_1,\dotsc,f_s\}$ such that $\deg_{X_i}(f_t)-\val_{X_i}(f_t)\le k\delta$.
\Statex
\State $\sigma\gets$ permutation such that $j\mapsto \alpha_{i,\sigma(j)}$ is non-decreasing;
\State $t\gets 1$;
\State $j_m\gets 1$;
\For{$\ell=1$ to $k-1$}
    \If{$\alpha_{i,\sigma(\ell+1)}-\alpha_{i,\sigma(\ell)}>\delta$}
        \State $f_t\gets\sum_{j=j_m}^\ell c_{\sigma(j)} X_1^{\alpha_{1,\sigma(j)}}\dotsb X_n^{\alpha_{n,\sigma(j)}}$;
        \State $j_m\gets\ell+1$;
        \State $t\gets t+1$;
    \EndIf
\EndFor
\State $f_t\gets\sum_{j=j_m}^k c_{\sigma(j)} X_1^{\alpha_{1,\sigma(j)}}\dotsb X_n^{\alpha_{n,\sigma(j)}}$;
\State \textbf{return} $\{f_1,\dotsc,f_t\}$. 
\end{algorithmic}
\end{algorithm}

\begin{algorithm}[htbp]
\caption{$\protect\Call{MultivariatePartition}{f,d_1,\dotsc,d_n}$}
\label{algo:multivariatepartition}
\begin{algorithmic}[1]
\Input $f\in\PolyRing$, $d_1$, \dots, $d_n\in\NN$;
\Output $\{f_1,\dotsc,f_s\}$.
\Statex
\State $\Scal\gets\{f\}$;
\State $s\gets 0$;
\While{$s<|\Scal|$}
    \State $s\gets|\Scal|$;
    \For{$i=0$ to $n$}
        \ForAll{$h\in\Scal$}
            \State $k\gets$ number of terms of $h$;
            \State $\delta\gets\gamma(k)d_i\max_{i'\neq i} (d_{i'})$;
            \State $\Scal_h\gets\Call{UnivariatePartition}{f,\delta,i}$;
        \EndFor
        \State $\Scal\gets\bigcup_{h\in\Scal} \Scal_h$;
    \EndFor
\EndWhile
\State \textbf{return} $\{ h/\XX^{\mval(h)} : h\in\Scal\}$.
\end{algorithmic}
\end{algorithm}

\begin{thm}
If $f\in\PolyRing$ is an $n$-variate $k$-nomial and $d_1$, \dots, $d_n$ are positive integers, the algorithm $\Call{MultivariatePartition}{f,d_1,\dotsc,d_n}$ runs in time $O(nk^2\log D)$ and outputs a partition $f_1+\dotsb+f_s$ of $f$ such that each $f_t$ has degree at most $\bigO(d^4k^3)$ in each variable where $d=\max_i(d_i)$ and for any multidimensional polynomial $g\in\PolyRing$ of multidegree at most $(d_1,\dotsc,d_n)$, 
\[\mult_g(f)=\min_{1\le t\le s}(\mult_g(f_t))\text.\]
\end{thm}

\begin{proof} 
The correctness of the algorithm and the degree bound follow from Lemma~\ref{lem:univariatepartition}. For the complexity, note that at each iteration of the while loop, the size of $\Scal$ increases by at least $1$, and the final size is bounded by $k$. This proves that this loop terminates in at most $k$ iterations. The global complexity follows from the complexity of \textsc{UnivariatePartition}, given in Lemma~\ref{lem:univariatepartition}.
\end{proof} 

\clearpage

\makeatletter
\renewcommand\@biblabel[1]{#1.}
\makeatother

\end{document}